\documentclass[runningheads]{llncs}

\usepackage{xspace}
\usepackage{amsmath}
\usepackage{amssymb}
\usepackage{amsfonts}
\usepackage{latexsym}
\usepackage{todonotes}
\usepackage{booktabs}
\usepackage{enumerate}
\usepackage[linesnumbered,ruled,noline,noend]{algorithm2e}
\usepackage[nodate]{datetime}

\newcommand{\ALG}{\ensuremath{\operatorname{\textsc{Alg}}}\xspace}
\newcommand{\OPT}{\ensuremath{\operatorname{\textsc{Opt}}}\xspace}
\newcommand{\opt}{\ensuremath{\operatorname{\textsc{Opt}}}\xspace}

\newcommand{\algis}{\ensuremath{\operatorname{Algorithm~\ref{alg:IS}}}\xspace}

\newcommand{\LA}{Late Accept\xspace}
\newcommand{\LR}{Late Reject\xspace}
\newcommand{\LAthenR}{Late Accept/Reject\xspace}
\newcommand{\LAM}{\LA model\xspace}
\newcommand{\LRM}{\LR model\xspace}
\newcommand{\LAthenRM}{\LAthenR model\xspace}
\newcommand{\SM}{standard model\xspace}

\newcommand\eps{\varepsilon}
\newcommand{\CEIL}[1]{\left\lceil#1\right\rceil}

\newcommand{\SET}[1]{\left\{#1\right\}}

\begin{document}

\title{Relaxing the Irrevocability Requirement for Online Graph Algorithms\thanks{Supported in part by the Danish Council for Independent Research, Natural Sciences, grant DFF-1323-00247, and the Villum Foundation, grant VKR023219.} 
}
\titlerunning{Relaxing Irrevocability} 
\author{     Joan Boyar \inst{1}
        \and Lene M. Favrholdt \inst{1}
        \and Michal Kotrb\v{c}\'{\i}k \inst{2}
        \and Kim S. Larsen \inst{1}
        }
\authorrunning{J. Boyar, L.\,M. Favrholdt, M. Kotrb\v{c}\'{\i}k, K.\,S. Larsen} 
\institute{University of Southern Denmark, Odense, Denmark,
\email{ \{joan,lenem,kslarsen\}@imada.sdu.dk}
\and
{The University of Queensland, Brisbane, Australia,
\email{m.kotrbcik@uq.edu.au}}
}

\maketitle

\begin{abstract}
Online graph problems are considered
in models where the irrevocability
requirement is relaxed. Motivated by practical examples where, for
example, there is a cost associated with building a facility and no
extra cost associated with doing it later, we consider the \LAM, where
a request can be accepted at a later point, but any acceptance is
irrevocable. Similarly, we also consider a \LRM, where an accepted
request can later be rejected, but any rejection is irrevocable (this
is sometimes called preemption).
Finally, we consider the \LAthenRM, where late accepts and rejects
are both allowed, but any late reject is irrevocable. 
For Independent Set, the \LAthenRM is
necessary to obtain a constant competitive ratio, but for Vertex
Cover the \LAM is sufficient and for Minimum Spanning Forest the
\LRM is sufficient. The Matching problem has a competitive ratio of $2$, but
in the \LAthenRM, its competitive ratio is~$\frac{3}{2}$.
\end{abstract}

\section{Introduction}
For an online problem, the input is a sequence of requests.
For each request, the algorithm has to make some decision without any
knowledge about possible future requests.
Often (part of) the decision is whether to accept or reject the
request and the decision is usually assumed to be irrevocable.
However, many online problems have applications for which total
irrevocability is not inherent or realistic.
Furthermore, when analyzing the quality of online algorithms,
relaxations of the irrevocability constraint often result in
dramatically different results, especially for graph problems.
This has already been realized and several papers study various
relaxations of the irrevocability requirement.
In this paper we initiate a systematic study of the nature of
irrevocability and of the implications for the performance of the
algorithms. 
Our aim is to understand whether it is the absence of
knowledge of the future or the irrevocability restrictions
on the manipulation of the solution set that makes an
online problem difficult.

We consider graph problems and focus on four classical problems,
\emph{Independent Set}, 
\emph{Matching}, \emph{Vertex Cover}, 
and \emph{Minimum Spanning Forest}.
Independent Set and Vertex Cover are studied in the vertex arrival
model.
In this model, vertices arrive one by one together with all the edges
between the newly arrived vertex and previous vertices.
Matching and Minimum Spanning Forest are studied in the edge arrival
model, but the results hold in the vertex arrival model as well.
In the edge arrival model, edges arrive one by one, and if a vertex
incident with the newly-arrived edge was not seen previously, it is
also revealed.

\subsection*{Relaxed irrevocability}

For the four problems considered in this
paper, the online decision is whether to accept or reject the current
request.
In the {\em standard} model of online problems, this decision is
irrevocable and has to be made without any knowledge about possible
future requests.
We relax the irrevocability requirement by allowing the 
algorithm to perform two additional operations, namely
\emph{late accept} and \emph{late reject}. 
Late accept allows the algorithm to accept not only the current request
but also requests that arrived earlier.
Thus, late accept relaxes irrevocability by not forcing the algorithm
to discard the items that are not used immediately. 
Late reject allows the algorithm to remove
items from the solution being constructed,
relaxing the irrevocability of the decision to accept an item.
When the algorithm is allowed to perform late accept or late reject,
but not both, we speak of a \emph{\LAM} and \emph{\LRM}, respectively.
Note that, in these two models, the late operations are irrevocable. 
We also consider the situation where the
algorithm is allowed to perform \emph{both} late accepts and late rejects,
focusing on the \emph{\LAthenRM}, where any item can be late-accepted
and late-rejected, but once it is late-rejected, this
decision is irrevocable.  In other words, if the algorithm performs
both late accept and late reject on a single item, the late accept has
to precede the late reject.  

We believe that the \LA, \LR, and \LAthenR models are
appropriate modeling tools corresponding to many natural settings.
Matching, for example, in the context of online gaming or chats,
functions in the \LAM. Indeed, the users are in the
pool until assigned, allowing the late accept, but once the users are
paired, the connection should not be broken by the operator.  
Note that the matching problem is a maximization problem.
For minimization problems, accepting a request may correspond to
establishing a resource at some cost.
Often there is no natural reason to require the establishment to
happen at a specific time.
Late acceptance was considered
for the dominating set problem in~\cite{BEFKL16}, which also contains
further feasible practical applications and additional rationale behind 
the model.

When the knapsack problem is studied in the \LRM,
items are usually 
called \emph{removable}; see for example~\cite{IT02,HKMG14,HKM15,CJS16,HM16}. 
For most other
problems, late rejection is usually called~\emph{preemption} and has been
studied in variants of many online problems,
for example
call control~\cite{BFL96,GGKMY97},
maximum coverage~\cite{SG09,RR16}, and
weighted matching problems~\cite{ELMS11,ELSW13}. Preemption was also previously
considered for one of the problems we consider here, independent set,
in~\cite{KKKK16}, but
in a model where advice is used, presenting lower bounds on
the amount of advice necessary to achieve given
competitive ratios in a stated range.

Online Vertex Cover was studied in~\cite{DP05}, where they
considered the possibility of swapping some of the
accepted vertices for other vertices at the very end, at some
cost depending on the number of vertices involved.

A similar concept is studied in, for example,~\cite{IW91,MSVW16,GGK16,GK14}
for online Steiner tree problems, MST, and TSP.
Here, replacing an accepted edge with another is allowed, 
and the objective
is to minimize the number of times this occurs while obtaining a good
competitive ratio. The problem is said to allow \emph{rearrangements} or 
\emph{recourse}.

TSP has also been studied~\cite{JL14} in a model where the actual acceptances 
and rejections (rejections carry a cost) are made at any time.

\subsection*{Competitive analysis}

For each graph problem, we study online algorithms in the standard,
\LA, \LR, and \LAthenR models using the standard tool of
{competitive analysis}~\cite{ST85j,KMRS88j},
where the performance of an online algorithm
is compared to the optimum algorithm \OPT via
the competitive ratio.
For any algorithm (online or offline), $A$, we let $A(\sigma)$ denote the
value of the objective function when $A$ is applied to the input
sequence~$\sigma$.

For minimization problems, we say that an algorithm, \ALG, is
 {\em $c$-competitive}, if there exists a constant
$\alpha$ such that, for all inputs $\sigma$, $\ALG(\sigma) \le c \cdot \OPT(\sigma) +
\alpha$.  
Similarly, for maximization problems, \ALG is $c$-competitive,
if there exists a constant $\alpha$ such that, for all inputs $\sigma$,
$\OPT(\sigma) \le c \cdot \ALG(\sigma) + \alpha$.
In both cases, if the inequality holds for $\alpha = 0$, the algorithm
is {\em strictly} $c$-competitive.  
The \emph{(strict) competitive ratio} of \ALG is the infimum over
all $c$ such that \ALG is (strictly) $c$-competitive.
The competitive ratio of a problem $P$ is the infimum over the
competitive ratio of all online algorithms for the problem.
For all combinations of the problem and the model, we obtain matching lower
and upper bounds on the competitive ratio.

For ease of notation for our results, we adopt the following conventions
to express that a problem essentially has competitive ratio~$n$,
i.e., it is true up to an additive constant.
We say that a problem has competitive ratio $n-\Theta(1)$ if
\begin{itemize}
\item for any algorithm, there is a constant $b>0$ such that the
strict competitive ratio is at least $n-b$, and
\item for any constant $b$, there is a strictly $(n-b)$-competitive algorithm
for graphs with at least $b+1$ vertices.
\end{itemize}
Similarly, we say that a problem has competitive ratio $n/\Theta(1)$ if
\begin{itemize}
\item for any algorithm, there is a constant $b>0$ such that the strict
competitive ratio is at least $n/b$, and
\item for any constant $b$, there is an $n/b$-competitive algorithm for
graphs with at least $b$ vertices.
\end{itemize}
This notation is used in Theorems~\ref{thm:islam} and~\ref{rej-vc}.
For all other results, the upper bounds hold for the strict
competitive ratio.
For convenience, when stating results containing both an upper bound
on the strict competitive ratio and a lower bound on the competitive
ratio, we use the term ``competitive ratio'' even though the result
holds for the strict competitive ratio as well.

\subsection*{Our results}

The paper shows that for some problems the \LAM allows for algorithms 
with significantly better competitive ratios, while for others it is the \LRM
which does. For other problems, the \LAthenRM is necessary to get
these improvements. See Table~\ref{table:results}.
Note that only deterministic algorithms are considered, not randomized
algorithms.

Our results on Minimum Spanning Forest follow from previous results.
Thus, they are mainly included to give an example where late rejects
bring down the competitive ratio dramatically.
The technical highlights of the paper are 
the results for Independent Set in the \LAthenRM,
where, in Theorems~\ref{thm:larIS-alg} and~\ref{thm:larIS-lb}, we prove
matching lower and upper bounds of $3\sqrt{3}/2$ on the competitive
ratio.

\begin{table}[htb]
\centering
\caption{Competitive ratios of the four problems in each of
the four models. $W$ is the ratio of the largest weight to the smallest.\label{table:results}}
\begin{tabular}{l@{\hspace{1em}}c@{\hspace{1em}}c@{\hspace{1em}}c@{\hspace{1em}}cccl}
\toprule
Problem & Standard & Late Accept & Late Reject & \multicolumn{4}{c}{Late Accept/Reject} \\[1ex]
\midrule
Independent Set & $n-1$ &  $\frac{n}{\Theta(1)}$ & $\CEIL{\frac{n}{2}}$ & \hspace{1.5em} & $\frac{3\sqrt{3}}{2}$ & $\approx$ & $2.598$  \\[1ex]
Matching & $2$ & $2$ & $2$  & \hspace{1.5em} & $\frac{3}{2}$ &&\\[1ex]
Vertex Cover & $n-1$ & $2$ & $n-\Theta(1)$ & \hspace{1.5em} & $2$ &&\\[1ex]
Min.\ Spanning Forest & $W$ & $W$ & $1$ & \hspace{1.5em} & $1$ &&\\[1ex]
\bottomrule
\end{tabular}
\end{table}

We consider only undirected graphs $G=(V,E)$.
Throughout the paper, $G$ will denote the graph under consideration,
and $V$ and $E$ will denote its vertex and edge set, respectively.
Moreover, $n = |V|$ will always denote the number of vertices in $G$.
We use $uv$ for the undirected edge connecting vertices $u$ and $v$,
so $vu$ denotes the same edge.

\section{Independent Set}
An \emph{independent set} for a graph $G=(V,E)$ is a subset $I
\subseteq V$ such that no two vertices in $I$ are connected by an edge.
For the problem called Independent Set, the objective is to find an
independent set of maximum cardinality.
We consider online Independent Set in the vertex arrival model.

\begin{theorem}
For Independent Set in the \SM, the strict competitive ratio is~$n-1$.
\end{theorem}

\begin{proof}
For the upper bound, consider the greedy algorithm that accepts each
vertex, if possible.

For the lower bound, consider the following adversarial
strategy where independent vertices arrive until a vertex, $v$, is
accepted by the algorithm.  From this point on, the adversary presents
vertices with $v$ as their only neighbor.  Hence, the algorithm cannot
accept any 
further vertices, whereas \OPT rejects $v$ and accepts all
other $n-1$ vertices, yielding the result. 
\mbox{}\qed\end{proof}

We first show that allowing only late rejects helps, but only very
slightly. 

\begin{theorem}
For Independent Set in the \LRM, the strict competitive ratio is~$\CEIL{n/2}$.
\end{theorem}

\begin{proof}
For the lower bound, whenever there is at least one vertex $v$ in the current
independent set constructed by \ALG, the adversary presents a vertex
incident only to $v$. The only vertex which can be accepted when \ALG
rejects $v$ is the vertex which just arrived,
so \ALG will never have more than one accepted vertex.
On the other hand, the graph the adversary produces is 
bipartite, so \opt can accept at least half of the vertices. 
 
For the upper bound, consider the following algorithm \ALG: If the
presented vertex $v$ can be added to the independent set $I$ being
constructed, then accept it. Otherwise, if $v$ is adjacent to only one
vertex $u$ in $I$, then remove $u$ from $I$ and add $v$ to $I$. 

By definition, \ALG accepts the first vertex.
If \ALG ever has two accepted
vertices, it will also have at least two accepted vertices at the end, and
the result holds.
Otherwise, consider some vertex~$u$ accepted by \ALG.
If the adversary presented a vertex not adjacent to $u$, \ALG
would accept it without rejecting $u$, which would be a contradiction.
Thus, each vertex presented by the adversary is connected to the
unique vertex currently in $I$.
By definition of the algorithm, in every step, the currently accepted
vertex is rejected and the new one accepted.
Thus, considering the vertices in the order they are presented,
they form a path.
No algorithm can accept more than every second vertex from a path,
and since all vertices are on the path,
\OPT accepts at most $\CEIL{n/2}$ vertices, and the result follows.
\mbox{}\qed\end{proof}

We now show that while allowing late accepts helps further, it is
still not enough to obtain a finite (constant) competitive ratio.

\begin{theorem}
\label{thm:islam}
For Independent Set in the \LAM, the competitive ratio is~$n/\Theta(1)$.
\end{theorem}

\begin{proof}
For a given positive constant $c$, an algorithm which does not 
accept any
vertex until the presented graph has an independent set of size at least~$c$,
and then accepts any such set, is $n/c$-competitive (for a graph with
no independent set of size $c$, $c-1$ suffices for the additive constant).

For the lower bound, consider any algorithm, \ALG.
Let $I_k$ denote a sequence consisting of $k$ independent vertices.
Let $c$ be the smallest integer such that \ALG accepts at
least one vertex of $I_c$.  If there is no such $c$, then the
competitive ratio of \ALG is unbounded. Assume that $v$ is a vertex
accepted in $I_c$. When \ALG accepts $v$, the adversary extends $I_c$
to an arbitrarily long input by presenting vertices of degree one
adjacent only to~$v$. Clearly, \OPT can accept $n-1$ vertices of the
resulting graph, while \ALG can accept at most $c$ vertices.
\mbox{}\qed\end{proof}

The following two theorems show that, in the \LAthenRM, the optimal
competitive ratio for Independent set is $3\sqrt{3}/2$.
The upper bound comes from a variant of the greedy algorithm,
Algorithm~\ref{alg:IS}, 
rejecting a set of vertices if it can be replaced by a set at least
$\sqrt{3}$ as large.
The algorithmic idea is natural and has been used
before (with other parameters than $\sqrt{3}$) in~\cite{RR16,SG09}, for example.
Thus, the challenge lies in deciding the parameter and proving the
resulting competitive ratio.
Pseudocode for Algorithm~\ref{alg:IS} is given below.

For Algorithm~\ref{alg:IS}, we introduce the following notation.
Let $S$ be the current set of vertices that have been accepted and not
late-rejected. 
Let $R$ be the set of vertices that have been late-rejected, and let
$P$ denote the set $V - (R\cup S)$ of vertices that have not been
accepted (and, hence, not late-rejected).

For a set $U$ of vertices, let $N(U) = \cup_{v\in U}N(v)$,
where $N(v)$ is the neighborhood of a vertex $v$ (not including $v$).
We call a set, $T$, of vertices \emph{admissible} if all the following
conditions are satisfied:
\begin{enumerate}[1)]
\item \label{adm:is} $T$ is an independent set;
\item \label{adm:p} $T \subseteq P$;
\item \label{adm:size} $|T| \ge \sqrt 3 \, |N(T)\cap S|$.
\end{enumerate}

\begin{algorithm}[H]
\DontPrintSemicolon
\SetKwInOut{Output}{output}
\KwResult{Independent set $S$}
\BlankLine
	$S=\emptyset$\; 
	\While{a vertex $v$ is presented}{
		\eIf{$S\cup\SET{v}$ is independent}{
			$S = S \cup \SET{v}$ \label{accept} \;
		}{
  			\While{there exists an admissible set }{
                                Let $T$ be an admissible set
                                minimizing $|S \cap N(T)|$\;
				$S = (S - N(T)) \cup T$ \label{late-accept} \;
                         }
		}
	}
\caption{Algorithm for Independent Set in the \LAthenRM.}
\label{alg:IS}
\end{algorithm}

For the analysis of Algorithm~\ref{alg:IS}, we partition $S$ into the
set, $A$, of vertices accepted in 
line~\ref{accept} and the set, $B$, of vertices accepted in
line~\ref{late-accept}.
We let $O$ be the independent set constructed by \opt.
For any set, $U$, of vertices, we let $U^+ = U \cap O$ and $U^- = U - O$.
Thus, $O = P^+ \cup S^+ \cup R^+ = P^+ \cup A^+ \cup B^+ \cup R^+$.

\begin{lemma}
\label{lemma:IS}
When Algorithm~\ref{alg:IS} terminates, $|B| \ge (\sqrt{3}-1)|R|$.
\end{lemma}

\begin{proof}
Clearly, the inequality is true before the first vertex is presented.
Each time a set, $X$, of vertices is moved from $S = A \cup B$ to $R$, a set at
least $\sqrt{3}$ times as large as $X$ is added to $B$.
Thus, each time the size of $R$ increases by some number $x$, the size
of $B$ increases by at least $\sqrt{3}x-x$.
The result follows inductively.  
\mbox{}\qed\end{proof}

\begin{lemma}
\label{lemma:pplus}
When Algorithm~\ref{alg:IS} terminates, $|P^+| < \sqrt{3} \, |S^-|$.
\end{lemma}

\begin{proof}
When the algorithm terminates, there are no admissible sets.
This means, in particular, that $P^+$ is not admissible.
Trivially, $P^+$ does not violate~\ref{adm:p}).
Furthermore, since $P^+ \subseteq O$, it cannot
violate~\ref{adm:is}).
Thus, we conclude from~\ref{adm:size}) that
$|P^+| <  \sqrt{3} |N(P^+) \cap S| \le \sqrt{3} |S^-|$, 
where the last inequality follows from the fact that there are no
edges between $P^+$ and $S^+$, since $P^+ \cup S^+ \subseteq O$.
\mbox{}\qed\end{proof}

\begin{lemma}
\label{lemma:brminus}
When Algorithm~\ref{alg:IS} terminates, $|B^-|+|R^-| \ge \sqrt{3} \, |R^+|$.
\end{lemma}

\begin{proof}
Consider a set, $T$, added to $B$ in line~\ref{late-accept}.
Let $Q = N(T) \cap S$.
We prove that 
\begin{equation}
\label{eq:tminus}
|T^-| \geq \sqrt{3} |Q^+|
\end{equation}
If $|Q^+|=0$, this is trivially true.
Thus, we can assume that $Q^-$ is a proper subset of $Q$.
Since $T$ is admissible, it follows that
\begin{equation}
\label{eq:tq}
|T| \geq \sqrt{3} |Q|
\end{equation}
Note that $(S - Q^-) \cup T^+$ is independent, since $(S-Q) \cup T$ is
independent and there are no edges between $Q^+$ and $T^+$.
Since the algorithm chooses $T$ such that $|Q|$ is minimized, this means
that 
\begin{equation}
\label{eq:tplus}
|T^+| < \sqrt{3} |Q^-|
\end{equation}
Subtracting Ineq.~(\ref{eq:tplus}) from Ineq.~(\ref{eq:tq}),
we obtain Ineq.~(\ref{eq:tminus}).

Let $T_1, T_2, \ldots, T_k$ be all the admissible sets that are chosen
in line~\ref{late-accept} during the run of the algorithm, and let
$Q_1, Q_2, \ldots , Q_k$ be the corresponding sets that are removed
from $S$.
Then, $\cup_{i=1}^k T_i \subseteq B \cup R$, and
thus, $\cup_{i=1}^k T_i^- \subseteq B^- \cup R^-$.
Furthermore, $R = \cup_{i=1}^k Q_i$.
Hence, $$|B^-|+|R^-| \geq \sum_{i=1}^k |T_i^-| \geq 
\sum_{i=1}^k \sqrt{3} |Q_i^+| = \sqrt{3} |R^+|,$$
where the second inequality follows from Ineq.~(\ref{eq:tminus}).
\mbox{}\qed\end{proof}

\begin{lemma}
\label{lemma:brplus}
When Algorithm~\ref{alg:IS} terminates,
$|B^+|+|R^+| \le \frac{\sqrt{3}}{\sqrt{3}+1}|B^+| + \frac{\sqrt{3}}{2}|B|$.
\end{lemma}
\begin{proof}
Since $|B^+| = |B| - |B^-|$ and $|R^+| =  |R| - |R^-|$,
we obtain the following.
\begin{align*}
(\sqrt{3}+1)(|B^+|+|R^+|) \;
&  =  \; \sqrt{3} \big(|B^+|+|R^+|\big) + 
         1 \cdot \big(|B|+|R| - (|B^-|+|R^-|)\big) \\
& \le \; \sqrt{3} |B^+| + \sqrt{3} |R^+| + |B|+|R| - \sqrt{3}|R^+|, 
         \text{ by Lemma~\ref{lemma:brminus}}\\
& =   \; \sqrt{3}|B^+|+|B|+|R| \\
& \le \; \sqrt{3}|B^+|+|B|+ \frac{1}{\sqrt{3}-1}|B|,
         \text{ by Lemma~\ref{lemma:IS}} \\
&  =  \; \sqrt{3}|B^+| + \frac{\sqrt{3}}{\sqrt{3}-1}|B|
\end{align*}
The result now follows by dividing both sides by $\sqrt{3}+1$.
\mbox{}\qed\end{proof}

\begin{theorem}
\label{thm:larIS-alg}
For Independent Set in the \LAthenRM,
\algis is strictly $3\sqrt{3}/2$-competitive.
\end{theorem}

\begin{proof}
We prove that $|O| \le \frac{3\sqrt{3}}{2} |S|$, establishing the result.
\begin{align*}
\opt \;
&  =  \; |P^+|+|A^+|+|B^+|+|R^+| \\
& \le \; \sqrt{3}(|A^-| + |B^-|) + |A^+| + |B^+| + |R^+|, 
         \text{ by Lemma~\ref{lemma:pplus}} \\ 
& \le \; \sqrt{3}(|A| + |B^-|) + |B^+| + |R^+|, 
         \text{ since } |A|=|A^+|+|A^-|\\
& \le \; \sqrt{3}\left(|A| + |B^-| + \frac{1}{\sqrt{3}+1}|B^+| +
                       \frac{1}{2}|B| \right),
         \text{ by Lemma~\ref{lemma:brplus}}\\
& \le \; \sqrt{3}\left(|A| + |B| + \frac{1}{2}|B| \right),
         \text{ since } \frac{1}{\sqrt{3}+1} < 1\\
& \le \; \frac{3 \sqrt{3}}{2} (|A|+|B|),
         \text{ since } \frac12 |B| \le \frac12 (|A|+|B|)
\end{align*}
\mbox{}\qed\end{proof}

We prove a matching lower bound:

\begin{theorem}
\label{thm:larIS-lb}
For Independent Set in the \LAthenRM, the competitive ratio is at
least~$3\sqrt{3}/2$.
\end{theorem}

\begin{proof}
Assume that \ALG is strictly $c$-competitive for some $c>1$.
We first show that $c$ is at least $3\sqrt{3}/2$ and then lift the
strictness restriction. Assume for the sake of contradiction that
$c<3\sqrt{3}/2$.

Incrementally, we construct an input consisting of a collection of
bags, where each bag is an independent set.  Whenever a new vertex $v$
belonging to some bag $B$ is given, we make it adjacent to every
vertex not in $B$, except vertices that have been late-rejected by
\ALG. Thus, if \ALG accepts $v$, it cannot hold any vertex in any
other bag. This 
implies that the currently accepted vertices of \ALG always form a 
subset of a single bag, which we refer to as \emph{\ALG's bag}, and
this is the crucial invariant in the proof. 
We say that \ALG \emph{switches} when it rejects the vertices of its
current bag and accepts vertices of a different bag.

For the incremental construction, the first bag is special in the
sense that \ALG cannot switch to another bag. We discuss later when we
decide to create the second bag, but all we will need is that the
first bag is large enough.  From the point where we have created a
second bag, \ALG has the option of switching.  Whenever \ALG switches
to a bag, $B'$, we start the next bag, $B''$.  All
that this means is that the vertices we give from this point on and
until \ALG switches bag again belong to $B''$, and \ALG
never holds vertices in the newest bag.

Now we argue that as long as we keep giving vertices, \ALG will
repeatedly have to switch bag in order to be $c$-competitive.  Choose
some $\eps > 0$, let $B$ be \ALG's bag, $B'$ be the new bag, and $s$
be the number of vertices which are not adjacent to any vertices in
$B'$.  If \ALG has accepted $a$ vertices of $B$ after $(c+\eps)a - s$
vertices of the new bag $B'$ have been given, \ALG has to accept at
least one additional vertex to be $c$-competitive, since at this point 
\OPT could accept all of the vertices in $B'$ and $s$ additional vertices.  
Since $B'$ is the new bag, $B$ has reached its final size, so eventually 
\ALG will have to 
switch to a different bag.

For the proof, we keep track of relevant parts of the behavior of \ALG
using a tree structure.  The first bag is the root of the tree. Recall
that whenever \ALG switches to a bag, say $X$, we start a new bag $Y$. 
In our tree structure we make $Y$ a child of $X$.

Since \ALG is $c$-competitive and always holds vertices only from a
single bag $B$, the number $a$ of vertices held in $B$ satisfies $a\ge |B|/c$.
Since, by assumption, $c<3$, it follows that \ALG can accept and then
reject disjoint sets of vertices of $B$ at most twice, or
equivalently, that each bag in the tree has at most two children.  As
we proved above, \ALG will have to keep switching bags, so if we keep
giving vertices, this will eventually lead to leaves arbitrarily far
from the root.

Consider a bag $B_m$ that \ALG holds
after a ``long enough'' sequence has been presented.

Label the bags from the root to \ALG's bag by
$B_1,\ldots, B_m$, where $B_{i+1}$ is a child of $B_i$ for each
$i = 1,\ldots, m-1$.  Let $a_j$, $1 \leq j < m$, be the number of vertices
of $B_j$ held by \ALG immediately before it rejected already accepted
vertices from $B_j$ for the first time and let $a_m$ be the number of
vertices currently accepted in $B_m$.  Let $n_j=|B_j|$, $1 \leq j \leq m$.

Furthermore, for each $j$, if $j$ is even,
let $s_j = a_2 + a_4 + \cdots + a_j$, and if $j$ is odd,
let $s_j = a_1 + a_3 + \cdots + a_j$. 
Note that our choice of adjacencies between bags implies that
\OPT can hold at least $s_j$ vertices in bags $B_1, B_2,\ldots, B_j$.

Thus, 
just before \ALG rejects the vertices in $B_{j-1}$ (just before
the $n_j$th vertex of $B_j$ is given),
we must have 
$c a_{j-1} \geq n_j - 1 + s_{j-2}$, by the assumption that \ALG is
$c$-competitive.
We want to introduce the arbitrarily small $\eps$ chosen above and 
eliminate the ``$-1$'' in this inequality: Since \OPT can always hold the
$a_1$ vertices from the root bag, $c a_j \geq a_1$ must hold for all
$j$.  Since $a_1 \geq (n_1-1)/c$, we get that $a_j \geq (n_1-1) / c^2$.
Thus, at the beginning of the input sequence, we can keep giving
vertices for the first bag,
making $n_1$ large enough such that $a_j$ becomes large enough
that $\eps a_{j-1} \geq 1$.  This establishes
$(c + \eps) a_{j-1} \geq n_j + s_{j-2}$.  Trivially, $n_j \geq a_j$, so
\begin{equation}
\label{eq:aj-1_aj}
(c + \eps) a_{j-1} - s_{j-2} \geq a_j.
\end{equation}
Next, we want to show that for any $1\leq c<3\sqrt{3}/2$, there exists an $m$
such that
\begin{equation}
\label{eq:opt}
s_m > ca_m,
\end{equation}
contradicting the assumption that $\ALG$ was $c$-competitive.
To accomplish this, we repeatedly strengthen Ineq.~(\ref{eq:opt}) by
replacing $a_j$ with the bound from Ineq.~(\ref{eq:aj-1_aj}),
eventually arriving at an inequality which can be proven to hold, and
then this will imply all the strengthened inequalities and, finally,
Ineq.~(\ref{eq:opt}).

From Ineq.~\ref{eq:opt}, we first use the definition of $s_m$
and collect the $a_m$ terms to get
\begin{equation}
\label{eq:om-2}
s_{m-2} > (c-1)a_m
\end{equation}
and then we use Ineq.~(\ref{eq:aj-1_aj}) to obtain the following
inequality, which implies Ineq.~(\ref{eq:om-2}):
\begin{equation*}
s_{m-2} > (c-1)(c+\eps)a_{m-1} - (c-1)s_{m-2}.
\end{equation*}
Collecting $s_{m-2}$ terms gives
\begin{equation*}
cs_{m-2} > (c-1)(c+\eps)a_{m-1},
\end{equation*}
and, using Ineq.~(\ref{eq:aj-1_aj}) again, we get a stronger inequality 
\begin{equation*}
cs_{m-2} > (c-1)(c+\eps)^2a_{m-2} - (c-1)(c+\eps)s_{m-3},
\end{equation*}
and, after moving $s_{m-3}$,
\begin{equation}
cs_{m-2} + (c-1)(c+\eps)s_{m-3} > (c-1)(c+\eps)^2a_{m-2}.
\end{equation}
We proceed by labeling the coefficients of $s_j$ and $a_j$
in these inequalities of the form
\begin{equation}
\label{generalized}
f_is_{m-i} + f_{i+1}s_{m-(i+1)} \ge g_ia_{m-i}
\end{equation}
by sequences $\SET{f_k}_{k=0}^{\infty}$, respectively
$\SET{g_k}_{k=0}^{\infty}$. We reverse the order of $f$ and $g$ so the
index $k$ corresponds to $k$ applications of Ineq.~(\ref{eq:aj-1_aj}),
and simplifications as in the above.  Therefore, $f_k$ and $g_k$ are
the coefficients of $s_{m-k}$ and $a_{m-k}$, respectively
(recall that $m$ is fixed).

Our repeated rewriting of $a_j$ and $s_j$ terms to terms with smaller
indices will eventually lead to $a_1$, which is a constant, and to
$s_0$, which is zero. The above calculations show that
\begin{eqnarray*}
f_0 & = & 1 \\
f_1 & = & 0 \\
f_2 & = & c \\
f_3 & = & (c-1)(c+\eps) \\
g_0 & = & c \\
g_1 & = & (c-1)(c+\eps) \\
g_2 & = & (c-1)(c+\eps)^2
\end{eqnarray*}
Our aim is to show that the coefficients $f_k$ and $g_k$ satisfy
\begin{eqnarray}
\label{eq:rec-g}
g_{i+1} & = & (c+\eps)(g_i - f_i) \\
\label{eq:rec-f}
f_{i+2} & = & g_i 
\end{eqnarray}
With the derived constants for $f$ and $g$ with small indexes given
above as the base case, we proceed by induction, assuming that for
$i$, the coefficients of Ineq.~(\ref{generalized}) have the values
claimed in Eq.~(\ref{eq:rec-g}) up to index $i$ and in
Eq.~(\ref{eq:rec-f}) up to index $i+1$.

We emphasize that we are still strengthening inequalities, so it is
the inequality for the $(i+1)$-version of Ineq.~(\ref{generalized})
that we derive that will imply Ineq.~(\ref{generalized}) for $i$.  The
induction is used to show that the coefficients in the inequalities
fulfill the recurrence equations stated for them.

From Ineq.~(\ref{generalized}), we collect the $a_{m-i}$ to obtain
$$
f_is_{m-(i+2)} + f_{i+1}s_{m-(i+1)} \ge (g_i-f_i)a_{m-i},
$$
and using Ineq.~(\ref{eq:aj-1_aj}), we get the stronger inequality 
$$
f_is_{m-(i+2)} + f_{i+1}s_{m-(i+1)} \ge (g_i-f_i)(c+\eps)a_{m-(i+1)} - (g_i-f_i)s_{m-(i+2)},
$$
which can be rewritten as
$$
g_is_{m-(i+2)} + f_{i+1}s_{m-(i+1)} \ge (g_i-f_i)(c+\eps)a_{m-(i+1)}.
$$
Reordering on the left-hand side and inserting the claimed
Eqs.~(\ref{eq:rec-g}) and~(\ref{eq:rec-f}), we arrive at
$$
f_{i+1}s_{m-(i+1)} + f_{i+2} s_{m-(i+2)} \ge g_{i+1} a_{m-(i+1)},
$$
which is the $(i+1)$-version of Ineq.~(\ref{generalized}),
proving the claim.

This concludes the strengthening of the inequalities and the sequence
of implications. Thus, to prove Ineq.~(\ref{eq:opt}), it is sufficient
to prove that there exist $m$ and $i$ such that
Ineq.~(\ref{generalized}) holds.  This, in turn, will follow if $g_j$
is negative. Indeed, if $g_j$ eventually becomes negative, we can
choose $j$ to be the smallest such index and Ineq.~(\ref{generalized})
then implies the desired Ineq.~(\ref{eq:opt}).  This inequality
contradicts the assumption of the algorithm being $c$-competitive, and
we will have established the theorem.

Plugging Eq.~(\ref{eq:rec-f}) into Eq.~(\ref{eq:rec-g}) gives us
\begin{equation}
\label{eq:rec}
g_{i+1} = (c+\eps)(g_i - g_{i-2} ),
\end{equation}
and this is the recurrence we use to show that $g_i$ becomes negative.

According to \cite[Theorem~1.2.1]{KL93},
every solution to a linear homogeneous difference equation
oscillates (around zero, and thus has negative values infinitely
often) if and only if
its characteristic equation has no positive roots.
For a direct proof of the subcase of the above theorem that we
need, see~\cite{LPS89}.

The characteristic equation of Eq.~\ref{eq:rec} is
$r^3 - (c+\eps)r^2 + (c+\eps)=0$,
which, for the interval $1\leq c+\eps<\frac{3\sqrt{3}}{2}$,
has two imaginary roots and one real root.
Letting $d=c+\eps$ and
\[s=-\sqrt[3]{108d - 8d^3 - 12\sqrt{-12d^4+81d^2}},\]
this third root is
$\frac{s}{6} + \frac{2d^2}{3s} + \frac{d}{3}$,
which is negative for $1\leq d\leq\frac{3\sqrt{3}}{2}$.
(Note that $s$ is no longer real when $-12d^4+81d^2$ becomes
negative, which occurs for $d> \frac{3\sqrt{3}}{2}$. At
$d=\frac{3\sqrt{3}}{2}$, in addition to the negative real
root, there is a double root at $\sqrt{3}$, and the solution
to the recurrence never becomes negative.)

Thus, from~\cite[Theorem~1.2.1]{KL93}, any solution to the
recurrence equation oscillates, implying, in particular,
that it becomes negative at some point, giving the desired contradiction.

Finally, we return to the assumption of strictness, which can easily
be removed.  There are only two places we use the relation given by
\ALG being strictly $c$-competitive. One place is in the claim
$c a_j \geq a_1$. However, we use this only to lead to
Ineq.~(\ref{eq:aj-1_aj}), and just as we used a large enough first bag
to make $\eps a_{j-1} \geq 1$, we can increase the size of the first
bag to eliminate any additive constant.  The other place was in the
argument that 
$c a_{j-1} \geq n_j - 1 + s_{j-2}$.
Also in this
case, if bags are large enough, no additive constant makes a
difference, and the minimum size of all bags can be increased by
increasing $n_1$, since that increases the lower bound on $a_1$, and
the algorithm can never hold fewer vertices in any bag than that.
\mbox{}\qed\end{proof}

The previous two theorems give us:
\begin{corollary}
For Independent Set in the \LAthenRM, the competitive ratio is~$3\sqrt{3}/2$.
\end{corollary}

\section{Matching} 
A \emph{matching} in a graph $G=(V,E)$ is a subset of $E$ consisting
of pairwise non-incident edges.
For the problem called Matching, the objective is to find a matching
of maximum cardinality. 
We study online Matching in the edge arrival model, but note that the
results hold in the vertex arrival model as well:
For the upper bounds, an algorithm in the vertex arrival model
can process the edges incident to the current vertex in any order.
For the lower bounds, all adversarial sequences used in this section
consist of paths, and hence, exactly the same input can be given in the vertex
arrival model.

It is well known and easy to prove that the greedy algorithm which adds an edge to the
matching whenever possible is $2$-competitive and this is optimal in
the \SM. The first published proof of this 
is perhaps in the classical paper of Korte and Hausmann~\cite{KH78}. The paper shows that in any graph,
the ratio of the minimum
size of a maximal matching to the size of a maximum matching
is at least~$\frac{1}{2}$, and there are graphs where it is no more
than~$\frac{1}{2}$.
Since the greedy algorithm produces a maximal matching, the claim follows.

Late accept or late reject alone does not help:

\begin{theorem}
For Matching in the \LAM, the competitive ratio is $2$.
\end{theorem}

\begin{proof}
The upper bound follows from the \SM. For the lower bound, 
we can use the same sequence as in the \SM:
The adversary presents $m$ mutually non-incident
edges to some algorithm, \ALG. For every
edge $uv$ accepted by \ALG at any point, the adversary presents
edges $xu$ and $vy$, which \ALG cannot accept.
Thus, there will be $m$ connected components such that
in each component, \OPT accepts at least twice as many edges as \ALG.
\mbox{}\qed\end{proof}

\begin{theorem}
For Matching in the \LRM, the competitive ratio is $2$.
\end{theorem}

\begin{proof}
The upper bound follows from the \SM. For the lower
bound, the adversary presents $m$ mutually
non-incident edges to some algorithm, \ALG. 
For each edge, $uv$, accepted by \ALG, the adversary presents an edge
$vx$.
If \ALG late-rejects $uv$, then the adversary presents an edge
$xy$. 
\ALG can only accept one of the edges $vx$ and $xy$ and it
cannot accept $uv$ again, but \OPT accepts both $uv$ and $xy$.
Otherwise, \ALG must reject $vx$.
In this case, the adversary presents an edge, $zu$.
\ALG can only keep $zu$ or $uv$, but \OPT accepts both $zu$ and $vx$.
This adversarial strategy results in $m$ connected components such that
in each component, \OPT accepts at least twice as many edges as \ALG.
\mbox{}\qed\end{proof}

\begin{theorem}
For Matching in the \LAthenRM, the  competitive ratio is at least $3/2$.
\end{theorem}

\begin{proof}
The adversary presents a number of mutually non-incident edges
to some algorithm, \ALG.
If, for some such edge $uv$, during the entire processing of the input,
\ALG does not accept $uv$, then \OPT will, and no more edges
incident to $u$ or $v$ will be presented.
The ratio is then unbounded on the subconstruction containing $uv$.

If \ALG accepts an edge $uv$, the adversary presents $xu$ and $vy$.
If \ALG never rejects $uv$, no more edges incident to any of these vertices
will be presented, and the ratio is $2$ on this subconstruction.

If \ALG late-rejects $uv$ at some point, then the adversary presents
$x'x$ and $yy'$. The algorithm cannot accept $uv$ again,
so it cannot accept more than two edges from this subconstruction,
while \OPT can accept three,
giving a ratio of~$3/2$.
\mbox{}\qed\end{proof}

To prove a matching upper bound, we give an algorithm,
Algorithm~\ref{alg:matching}, which is strictly $\frac{3}{2}$-competitive
in the \LAthenRM. 

Recall that for a matching $M$, a path $P=e_1,\ldots, e_k$ is
{\em alternating} with respect to $M$, if for all $i\in \{1,\ldots, k\}$,
$e_i$ belongs to $M$ if and only if $i$ is even. 
Moreover, an alternating path $P$ is called {\em augmenting} if neither endpoint
of $P$ is incident to a matched edge.
Note that the symmetric difference of a matching $M$ and
an augmenting path with respect to $M$
is a matching of size larger than $M$.
We focus on local changes, called short augmentations in~\cite{VH05}.
We use a result which implies
that if a maximal matching $M$ does not admit augmenting paths of
length~$3$, then $3|M| \ge 2|OPT|$. This fact is part of the folklore 
and its proof can be found for example in~\cite[Lemma~2]{FKMSZ05}.

\begin{algorithm}[H]
\caption{Algorithm for maximal matching in the \LAthenRM.}
\label{alg:matching}
\DontPrintSemicolon
\KwResult{Matching $M$}
	$M=\emptyset$\;
	\While{an edge $e$ is presented}{
		\If{$M\cup\SET{e}$ is a matching}{
			$M = M \cup\SET{e}$\;
		}{
  			\If{there is an 
                           augmenting path $xuvy$ of length~$3$ }{
				$M = (M\cup\SET{ux,vy}) - \SET{uv}$\;
			}
		}
	}
\end{algorithm}

\begin{theorem}
\label{thm:larM}
For Matching in the \LAthenRM, Algorithm~\ref{alg:matching}
is strictly $3/2$-competitive.
\end{theorem}

\begin{proof}
We first show that Algorithm~\ref{alg:matching} is a \LAthenR
algorithm, i.e., no edge which is late-rejected is later late-accepted. 
Suppose an edge $e=uv$ is late-rejected in some step $i$. For it to be
accepted again, there must later be an augmenting path consisting of three
edges, where $e$ is one of the outer edges, and one endpoint of $e$ must
not be incident to any edges of the matching at that point. However, once
a vertex is incident to an edge in some matching, it is always incident
to some edge in all later matchings, due to the augmentation. Thus, $e$ cannot
be late-accepted later, so after any edge is late-rejected, it will never
be accepted again.

Next, we prove that the algorithm is strictly $3/2$-competitive.
Let $M$ be the matching constructed by Algorithm~\ref{alg:matching} on
a graph~$G$. Algorithm~\ref{alg:matching} ensures
that $G$ does not contain any augmenting paths of length at
most three with respect to $M$ by augmenting on them when they do exist.  
To prove that this fact implies the bound, we present a proof similar
to~\cite[Lemma~2]{FKMSZ05}. Let $M'$ be any maximum matching in~$G$.
Consider the symmetric difference $N$ of $M$ and $M'$.
Since $M$ and $M'$ are matchings,
any path in $N$ contains alternatingly edges from $M$ and $M'$.
Since each augmenting path with respect to $M$ in $N$ contains one more
edge of $M'$ than of $M$, we get that $N$ contains $|M'| - |M|$
augmenting paths.  Clearly, no augmenting path
in $N$ consists of a single edge, since 
all edges which either are not accepted or are late-rejected by
Algorithm~\ref{alg:matching} are incident to at least one edge accepted
by the algorithm.
Since there is no augmenting path of length at most three with respect
to $M$ in $G$, it follows that all of the $|M'| - |M|$
augmenting paths in $N$ have length at least five.  At least two
edges of an augmenting path in $N$ of length at least five are contained in
$M$, and thus accepted by the algorithm, so
we get that $|M| \geq 2(|M'| - |M|)$, giving that $|M'| \leq \frac{3}{2}|M|$.
\mbox{}\qed\end{proof}

\section{Vertex Cover}
A \emph{vertex cover} for a graph $G=(V,E)$
is a subset $C\subseteq V$ such that for any edge,
$uv\in E$, $\SET{ u,v} \cap C \not= \emptyset$.
For the problem called Vertex Cover, the objective is to find a vertex
cover of minimum cardinality.
We study online Vertex Cover in the vertex arrival model.

\begin{theorem}
\label{stdvc}
For Vertex Cover in the \SM, the strict competitive ratio is~$n-1$.
\end{theorem}

\begin{proof}
For the lower bound, consider any online algorithm, \ALG.
For each $n$, the adversary presents independent vertices until \ALG
rejects some 
vertex or $n$ vertices have been presented. If $n$ vertices are presented,
\OPT accepts none of them and the ratio is unbounded. 
If the algorithm rejects some vertex $v$,
then the remainder of the $n$ vertices will be adjacent  
only to~$v$. \OPT will only accept $v$ and the result follows.

For the upper bound, the algorithm only accepts a new vertex $v$ if at
least one edge incident with $v$  
is not already covered. Thus, it rejects the first
vertex and therefore accepts at most $n-1$ vertices. \OPT accepts at least
one vertex unless there are no edges, in which case the algorithm
does not accept any vertices either.
\mbox{}\qed\end{proof}

The situation improves dramatically if we can accept vertices at a later stage.

\begin{theorem}
\label{avc}
For Vertex Cover in the \LAM, the competitive ratio is~$2$.
\end{theorem}

\begin{proof}
The best known offline $2$-approximation algorithm for Vertex Cover
greedily maintains a maximal matching, repeatedly covering both endpoints of
an edge and removing all edges incident to these two endpoints.
In the \LAM, a $2$-competitive online algorithm can be obtained by mimicking
the offline approximation algorithm. 
The online algorithm does not accept any
vertex until it sees the second vertex incident to an uncovered edge;
then it accepts both endpoints of that edge.

For the lower bound, consider any algorithm \ALG. The adversary presents 
isolated pairs of vertices, each pair connected by an edge. After the 
second vertex of a pair has arrived, \ALG must have accepted at least
one of them, 
or the adversary could stop the input there, and \ALG's output would not
be a vertex cover.
If \ALG accepts both vertices, then no further vertices adjacent to the pair
arrive, and \OPT could have covered the edge with only one vertex. If \ALG 
accepts only one vertex $u$ from a pair $\SET{ u,v }$, then an additional
vertex adjacent only to $v$ arrives, and \OPT could cover both edges
with only $v$, but \ALG must accept at least two of the three vertices.
\mbox{}\qed\end{proof}

Allowing both late accept and late reject does not improve the
situation further.

\begin{theorem}
For Vertex Cover in the \LAthenRM, the competitive ratio is~$2$.
\end{theorem}

\begin{proof}
The upper bound follows from Theorem~\ref{avc}.  The lower bound
follows from the observation that no algorithm that ever late-rejects
a vertex can be $c$-competitive for any constant~$c$.
Indeed, if \ALG late-rejects a vertex $v$, then the adversary
can present arbitrarily many vertices adjacent only to $v$.
Therefore, to be $c$-competitive for any constant $c$,
\ALG can never late-reject a vertex and the lower bound from
Theorem~\ref{avc} applies.
\mbox{}\qed\end{proof}

\begin{theorem}
\label{rej-vc}
For Vertex Cover in the \LRM, the competitive ratio is $n-\Theta(1)$.
\end{theorem}

\begin{proof}
For the lower bound, the adversary keeps giving independent vertices
until the algorithm rejects at least one vertex, $v$.
Since \OPT does not have to accept any vertices if they are all
independent, the algorithm must eventually reject at least one vertex
to avoid an unbounded competitive ratio.
We let $b$ denote the number of vertices presented at the time the
first vertex is rejected.
After this point, all new vertices are adjacent to $v$, so the
algorithm has to accept all of them.
In total, at least $n-b$ vertices are accepted, and \OPT accepts only~$v$.

For the upper bound, consider the following algorithm, $\ALG_b$:
The first $b+1$ vertices are accepted (if they arrive).
After that an optimal vertex cover, $C$, for the edges seen so far is
calculated. 
Each vertex not included in $C$ is rejected.
After this, each new vertex is accepted only if necessary.
Note that the size of $C$ is a lower bound on \OPT.
Thus, for any input sequence $I$ of length at least $b+1$, either
$\OPT(I) = \ALG(I) = 0$ or 
$$\frac{\ALG_b(I)}{\OPT(I)} = \frac{|C|+n-(b+1)}{\OPT(I)} 
  \leq \frac{\OPT(I)+n-(b+1)}{\OPT(I)} \leq n-b\,.$$
\mbox{}\qed\end{proof}

\section{Minimum Spanning Forest} 
A \emph{spanning forest} for a graph $G=(V,E)$
is a subset $T\subseteq E$ which forms a spanning
tree on each of the connected components of $G$.
Given a weight function $w \colon E \to \mathbb{R^+}$,
the objective of the Minimum Spanning Forest problem is to find
a spanning forest of minimum total weight.
We let $W$ denote the ratio between the largest and the smallest
weight of any edge in the graph.

We study online Minimum Spanning Forest in the edge arrival model, but
the results also hold in the vertex arrival model: 
For the upper bounds, an algorithm in the vertex arrival model
can process the edges incident to the current vertex in any order.
In the lower bound sequences presented here, all edges from a new vertex to
all previous vertices are presented together in an arbitrary order.

\begin{theorem}
For Minimum Spanning Forest in the \SM, the competitive ratio is~$W$.
\end{theorem}
\begin{proof}
Since all spanning forests have the same number of edges, the ratio
cannot be worse than $W$.
A matching upper bound can be realized by the adversary first
presenting a tree consisting of edges of weight~$W$, and then presenting 
edges of weight~$1$ from a new vertex $v$ to each of the vertices seen so far.
The ratio is $\frac{(n-2)W+1}{n-1}$, giving an asymptotic
lower bound of~$W$.
\mbox{}\qed\end{proof}

Since an online algorithm does not know when the input ends,
it must always have a forest spanning all the vertices seen so far,
so in moving from the \SM to the \LAM, we do not gain any advantage:

\begin{theorem}
For Minimum Spanning Forest in the \LAM, the competitive ratio is~$W$.
\end{theorem}

\begin{proof}
We show that we can never perform a late accept. Assume to the
contrary that an
edge $uv$ was late-accepted and added to a solution $F'$ for the
current graph $G'=(V',E')$.
Since $uv$ was late-accepted, both vertices $u$ and $v$
were seen earlier and thus contained in $V'$.  By our requirement
that the algorithms maintain a spanning forest on the set of vertices
presented so far, $F'$ is a spanning forest of $G'$.
Therefore, adding $uv$ created a cycle,
contradicting that the algorithm finds a forest.
\mbox{}\qed\end{proof}

On the other hand, in the \LRM, the greedy online algorithm mentioned by
Tarjan in~\cite{T83} can be used.
We detail the algorithm in the proof.

\begin{theorem}
For Minimum Spanning Forest in the \LRM, the competitive ratio is~$1$.
\end{theorem}

\begin{proof}
No algorithm can be better than $1$-competitive.
For the upper bound, we note that the greedy online algorithm
mentioned by Tarjan in~\cite{T83} works in the \LRM:
Assume that the current forest is~$F'$ when an edge $e=uv$ arrives.
If at least one of the two endpoints of $e$ is a vertex not seen earlier,
accept $e$. Otherwise,
the greedy algorithm constructs the unique cycle $C_e$
in $F' \cup \SET{e}$. If $e$ is not the heaviest edge in $C_e$, then the
algorithm late-rejects the heaviest edge $f$ in $C_e$ and replaces it
by $e$, obtaining $F''$. Otherwise, it rejects $e$.
It is easy to see that this produces an optimal
spanning forest; it only uses the so-called red rule~\cite{T83}.
\mbox{}\qed\end{proof}

Since the \LRM leads to an optimal spanning tree, any model allowing that
possibility inherits the result.

\begin{theorem}
For Minimum Spanning Forest in the \LAthenRM, the competitive ratio is~$1$.
\end{theorem}

\section*{Future Work}

Since we prove tight results for all combinations of problems and
models considered, 
we leave no immediate open problems.
However, one could reasonably consider late operations a resource
to be used sparingly, as for the rearrangements 
in~\cite{IW91,MSVW16,GGK16,GK14}, for example.
Thus, an interesting continuation of our work would be a study of
trade-offs between the number of late operations employed
and the quality of the solution (in terms of competitiveness).
Obviously, one could also investigate other online problems
and further model variations.

\bibliographystyle{plain}
\bibliography{refs}

\end{document}